\pdfoutput=1
\RequirePackage{ifpdf}
\ifpdf 
\documentclass[pdftex]{sigma}
\else
\documentclass{sigma}
\fi

\usepackage{booktabs}
\usepackage{mathrsfs}
\usepackage{tensor}

\numberwithin{equation}{section}

\newtheorem{Theorem}{Theorem}[section]
\newtheorem{Lemma}[Theorem]{Lemma}
\newtheorem{Proposition}[Theorem]{Proposition}
\newtheorem{Observation}[Theorem]{Observation}
 { \theoremstyle{definition}
\newtheorem{Definition}[Theorem]{Definition}
\newtheorem{Example}[Theorem]{Example}
\newtheorem{Remark}[Theorem]{Remark} }

\begin{document}

\newcommand{\arXivNumber}{2010.03638}

\renewcommand{\PaperNumber}{015}

\FirstPageHeading

\ShortArticleName{St\"ackel Equivalence of Non-Degenerate Superintegrable Systems, and Invariant Quadrics}

\ArticleName{St\"ackel Equivalence of Non-Degenerate\\ Superintegrable Systems, and Invariant Quadrics}

\Author{Andreas VOLLMER~$^{\dag\ddag}$}

\AuthorNameForHeading{A.~Vollmer}

\Address{$^\dag$~Institute of Geometry and Topology, University of Stuttgart, 70550 Stuttgart, Germany}
\Address{$^\ddag$~Dipartimento di Scienze Matematiche (DISMA), Politecnico di Torino,\\
\hphantom{$^\ddag$}~Corso Duca degli Abruzzi, 24, 10129 Torino, Italy}
\EmailD{\href{mailto:andreas.vollmer@polito.it}{andreas.vollmer@polito.it}, \href{mailto:andreas.d.vollmer@gmail.com}{andreas.d.vollmer@gmail.com}}

\ArticleDates{Received October 09, 2020, in final form February 02, 2021; Published online February~17, 2021}

\Abstract{A non-degenerate second-order maximally conformally superintegrable system in dimension~2 naturally gives rise to a quadric with position dependent coefficients. It is shown how the system's St\"ackel class can be obtained from this associated quadric.
The St\"ackel class of a second-order maximally conformally superintegrable system is its equivalence class under St\"ackel transformations, i.e., under coupling-constant metamorphosis.}


\Classification{14H70; 70H06; 30F45}

\section{Introduction}
Superintegrable systems in dimension~2 are a classical subject of mathematical research. They have interrelations, for instance, with special functions \cite{GIVZ2013, KMP13}, quadratic algebras \cite{Daskaloyannis2001,Kress07,Post2011}, degenerate Poisson pencils~\cite{MPT2011}, and constant divisors in the context of the Riemann--Roch theorem \cite{Tsiganov2020}, see also \cite{Tsiganov1999,Tsiganov2019}.
In particular, second-order (maximally) superintegrable systems in dimension 2 (2D) are classified \cite{Kalnins&Kress&Miller-I,Kalnins&Kress&Miller-II, Kalnins&Kress&Pogosyan&Miller,KPM2002}, for Euclidean 2-space even algebraic geometrically~\cite{Kalnins&Kress&Miller,Kress&Schoebel}.
In this context, superintegrable systems are considered up to isometries, i.e., locally up to coordinate transformations.

Another viewpoint is to consider second-order superintegrable systems up to St\"ackel (i.e., conformal) transformations~\cite{BKM1986}, also known as coupling-constant metamorphism~\cite{HGDR84}. This viewpoint appears, for instance, in \cite{Capel_phdthesis,DY06,Kress07,KSV_2020}. The equivalence classes resulting from this identification are called \emph{St\"ackel classes}.
The equivalence classes of non-degenerate second-order superintegrable systems under St\"ackel transforms are classified \cite{DY06,Kress07}. Specifically, in~\cite{Kress07} a method is developed that allows one to identify the St\"ackel class of a non-degenerate superintegrable system from the properties of its associated quadratic algebra, see also~\cite{Post2011}.
The present paper presents an alternative method to determine the St\"ackel class of a given non-degenerate second-order 2D superintegrable system. It also applies to conformally superintegrable systems.
This new method, presented below in Theorem~\ref{thm:classes}, is applicable to 2D non-degenerate second-order superintegrable systems of any curvature. Specifically we do not require constant sectional curvature.

On the other hand, we shall see that the proof for this new method is extremely simple, since we do not have to start our proof ``from scratch''. Instead, we shall exploit the \emph{known} flat second-order systems that realise the \emph{known} St\"ackel classes. With this in mind, the purpose of this paper is not to present a tedious, complicated proof, but rather to report an efficient tool that determines which St\"ackel class a given non-degenerate second-order 2D (conformally) superintegrable system belongs to.

\section{Preliminaries}

Let $g$ be a (pseudo-)Riemannian metric on a $2$-dimensional manifold $M$ and consider the Hamiltonian $H(x,p)=g^{ij}(x)p_ip_j+V(x)$. Here $x$ and $p$ stand, respectively, for position coordinates $x^i$ and canonical momenta (fibre coordinates) $p_i$ on the cotangent space $T^*M$.
Note that in what follows we shall consider two Hamiltonians $H_1$, $H_2$ to be equal if they are constant multiples as functions on $T^*M$.

A \emph{second-order integral} (of motion) for $H$ is a function $F(x,p)=K^{ij}(x)p_ip_j+W(x)$ such that~$H$ and~$F$ commute w.r.t.\ the canonical Poisson bracket on $M$ (Einstein's summation convention applies),
\begin{equation}\label{eq:integral}
	\{ F,H \}
	= \frac{\partial F}{\partial x^i}\frac{\partial H}{\partial p_i}
	- \frac{\partial H}{\partial x^i}\frac{\partial F}{\partial p_i}
	= 0 .
\end{equation}
We remark that condition~\eqref{eq:integral} is equivalent to requiring that the coefficients $K_{ij}$ (indices are lowered using $g$) form the components of a Killing tensor $K$ and that $K$ is compatible with $V$ according to the \emph{Bertrand--Darboux condition} \cite{bertrand_1857,darboux_1901}
\begin{equation}\label{eq:dKdV}
	\nabla_{[i}K\indices{_{j]}^a}\nabla_aV = K\indices{^a_{[i}}\nabla^2_{j]a}V,
\end{equation}
where square brackets denote antisymmetrisation; $\nabla$ is the Levi-Civita connection for $g$ and $\nabla^2$ denotes the Hessian.
Equation~\eqref{eq:dKdV} is the integrability condition for $W$; it is obtained as follows: First solve the homogeneous linear component (in momenta) of the polynomial equation~\eqref{eq:integral} for the differential~${\rm d}W$, obtaining $\nabla_aW=K_{ab}\nabla^bV$. Then take a further derivative; the integrability condition for~$W$,
i.e.~${\rm dd}W=0$, is~\eqref{eq:dKdV}.

\begin{Remark}Note that we work over the field $\mathbb{C}$ of complex numbers unless otherwise indicated.
\end{Remark}

\begin{Definition}
A (2D maximally) \emph{second-order superintegrable system} is a Hamiltonian $H$ together with its space of second-order integrals \eqref{eq:integral}, which are required to contain two elements~$F_1$ and~$F_2$ such that $(H,F_1,F_2)$ are functionally independent.
\end{Definition}

\begin{Remark}	Since from the knowledge of the superintegrable Hamiltonian we can reconstruct the system, we typically specify only the Hamiltonian.
\end{Remark}

Within the theory of superintegrable systems, \emph{non-degeneracy} is a well-known natural property \cite{Kalnins&Kress&Miller-I,KKM_2018,Kress&Schoebel}. In Definition~\ref{def:non-degeneracy}, we complement the standard definition~(i) \cite{Kalnins&Kress&Pogosyan&Miller,KPM2002,Kress&Schoebel&Vollmer} by a~helpful distinction in~(ii).

\begin{Definition}\label{def:non-degeneracy}\quad
\begin{enumerate}\itemsep=0pt
\item[(i)] A second-order superintegrable system in dimension~$n$ is \emph{non-degenerate} if equation~\eqref{eq:dKdV} admits a $(n+2)$-dimensional space $\mathcal{U}$ of solutions (for~$V$).
	
\item[(ii)] We can consider the Hamiltonian with a fixed potential $V\in\mathcal{U}$ or with the full $(n+2)$-parameter family. In the latter case we write~$V^{\mathcal{U}}$ for clarity. Note that~$V^{\mathcal{U}}$ denotes a~specific parametrisation of~$\mathcal{U}$.
\end{enumerate}
\end{Definition}

\begin{Remark}
	It is in place to remark on the system [E15] of \cite{Kalnins&Kress&Pogosyan&Miller}, which appears alongside the non-degenerate systems in \cite{Kress07}.
	The reader will notice that the system [E15] does not appear in Table~\ref{tab:classes}, while it appears in the corresponding table of \cite{Kress07}. This is correct, and a closer look confirms that the system [E15] does not meet the prerequisites of Definition~\ref{def:non-degeneracy}.
	Indeed, the potential of [E15] depends on a functional parameter, inconsistent with non-degeneracy. Note that the system [E15] is also not included in the algebraic-geometric classification of non-degenerate second-order superintegrable systems on flat 2-space, see \cite{Kress&Schoebel}.
\end{Remark}

The following equivalence relation of second-order maximally superintegrable systems is well known.

\begin{Definition}\label{def:staeckel}
	Let $\langle H,F_i\rangle$ be a second-order non-degenerate superintegrable system with the Hamiltonian~$H$ and integrals $F_i$.
	Let $U\in\mathcal{U}$ be one of its compatible potentials.
	Then $\big\langle\tilde H,\tilde F_i\big\rangle$ with
	\begin{equation}\label{eq:staeckel}
		\tilde H=U^{-1}H ,\qquad
		\tilde F_i = F_i+(1-W_i)U^{-1}H
	\end{equation}
	is called the \emph{St\"ackel transform} of $\langle H,F_i\rangle$.
\end{Definition}

Often, since $U\in\mathcal{U}$, lives in a linear space of admissible potentials, the St\"ackel transform is considered as a transformation involving a coupling parameter $\alpha$, i.e., the Hamiltonian $H_\alpha=H_\alpha(x,y)=H+\alpha U$ is taken as dependent on a linear parameter $\alpha$. Then, St\"ackel transform can be interpreted as a change of the Hamiltonian such that the roles of the energy and of the coupling parameter are interchanged \cite{Kalnins&Kress&Miller-II,Post10}. A multi-parameter version of this approach to St\"ackel transforms has been developed by Sergyeyev and Błaszak \cite{Sergyeyev_2008}, see also \cite{Blaszak_2012,Blaszak_2017}. In our definition, the coupling parameters are naturally geometrically replaced by the freedom to choose another element $U$ from the linear space $\mathcal{U}$.

\begin{Remark}\quad
\begin{enumerate}\itemsep=0pt
\item[(i)] The integrals $\tilde F_i$ are indeed integrals for $\tilde H$ and, provided they are functionally independent, the St\"ackel transform is again a second-order superintegrable system.
\item[(ii)] In the context of second-order superintegrable systems, the \emph{St\"ackel transform} is also known under the name \emph{coupling constant metamorphosis} \cite{BKM1986,HGDR84}, but these two concepts are not the same in other contexts~\cite{Post10}.
\end{enumerate}
\end{Remark}

\begin{Definition}\quad
\begin{enumerate}\itemsep=0pt
\item[(i)] We say that the non-degenerate second-order superintegrable systems $\langle H,F_i\rangle$ and $\big\langle \tilde H,\tilde F_i\big\rangle$ are \emph{St\"ackel equivalent} if they are St\"ackel transforms.
\item[(ii)] The equivalence class of a non-degenerate second-order superintegrable system under St\"ackel equivalence is called its \emph{St\"ackel class}~$\mathcal{S}$.
\end{enumerate}
\end{Definition}

Note that St\"ackel transform is ``conformal'' in the following sense: For two systems $h,\tilde h\in\mathcal{S}$, belonging to the same St\"ackel class, the underlying metrics $g_h$ and $g_{\tilde h}$ are conformally equivalent.

\section{Method}
The aim of the current section is to construct a certain variety $\mathcal{Q}$ that is invariant under St\"ackel transform. It is encoded in a quadric, for a given (non-degenerate) 2D second-order superintegrable system. In the following section we shall relate this variety to the set of all flat realisations of a given St\"ackel class, i.e., those non-degenerate 2D second-order superintegrable systems inside the given St\"ackel class that are realised on a flat geometry.
In order to prevent misunderstandings, it is worthwhile to state a clarification: The present paper seeks a method that determines the St\"ackel class of \emph{any} given 2-dimensional non-degenerate second-order superintegrable system, i.e., for manifolds of arbitrary (including non-constant) curvature. We base this method on an invariant variety~$\mathcal{Q}$ inside an invariant projective space~$\mathcal{W}$.
The variety~$\mathcal{Q}$ is defined by requiring flatness for objects in $\mathcal{W}$. We would like to highlight that this does not imply any restriction to flat manifolds as far as Theorem~\ref{thm:classes}, and Theorem~\ref{thm:conformal}, are concerned. It is merely a technical trick, i.e., a good choice, which enables us to reduce the problem under consideration to a simple computation.

\begin{Definition}
A \emph{quadric} in projective space $\mathbb{P}^m$ is the subset defined by the zero set of a~homogeneous quadratic polynomial equation in $m+1$ variables.
\end{Definition}

Note that we do not require the polynomial equation to be irreducible.
The following observation is the basis for the technique developed further below.
\begin{Observation}\label{obs:main}\looseness=-1
	Consider two St\"ackel equivalent, non-degenerate second-order superintegrable Hamiltonians $H=g^{ij}p_ip_j+V$
	and $\tilde H=\tilde g^{ij}p_ip_j+\tilde V$.
	By Definition~{\rm \ref{def:non-degeneracy}}, they give rise to $(n+2)$-dimensional spaces $\mathcal{U}$ and $\tilde{\mathcal{U}}$ solving~\eqref{eq:dKdV}, respectively.
	We have, according to Definition~{\rm \ref{def:staeckel}}, that
	\[
	\mathcal{U}g = \frac{\mathcal{U}}{U} Ug = \tilde{\mathcal{U}}\tilde g,
	\]
	where $U$ is a solution of~\eqref{eq:dKdV} for the superintegrable system arising from~$H$. This is true in any dimension $n\geq2$.
As a result, for a non-degenerate second-order superintegrable Hamiltonian~$H$, the $(n+2)$-dimensional linear space $\mathcal{V}$,
	\begin{equation*}
		\mathcal{V} = \mathcal{U}g ,
	\end{equation*}
	is invariant under St\"ackel transformations.
\end{Observation}

We comment that an analogous observation holds for conformally superintegrable Hamiltonians, but then $U$ can be any function (see details in Section~\ref{sec:conformal}).

Elements $g\in\mathcal{V}$ are typically metrics, since $\det(Ug)\ne0$ for $U\ne0$. Clearly the origin never corresponds to a metric. We observe that constant multiples of $U$ give rise to proportional Hamiltonians $H=\frac{1}{U}\big(g^{ij}p_ip_j+V^\mathcal{U}\big)$. It is therefore useful to reconsider the $(n+2)$-dimensional linear space $\mathcal{V}$ as an $(n+1)$-dimensional projective space, which we denote by
\[
\mathcal{W} = (\mathcal{V}\setminus\{0\})/{\sim},
\]
where $h\sim k$ for $h,k\in\mathcal{V}$ if $h=ak$ for constant $a\ne0$.

\begin{Remark}Note that while we shall perform the computations using a specific parametrisation $V^{\mathcal{U}}$ of the potential and particular coordinates on $M$, the space $\mathcal{W}$ is indeed a geometric object. It is independent of the specific parametrisation and the choice of coordinates.
\end{Remark}

\begin{Example}
	We emphasize that elements $q\in\mathcal{W}$ are (classes of) symmetric 2-tensors (and in fact metrics).
	For instance, take the 2-sphere with the round metric
	\begin{equation}\label{eqn:S4.metric}
		g = \frac{{\rm d}x{\rm d}y}{(xy+4)^2} .
	\end{equation}
	For the Hamiltonian we take that of the system [S4] in~\cite{Kalnins&Kress&Pogosyan&Miller}, i.e.,
	\begin{equation}\label{eqn:S4.hamilton}
		H = (xy+4)^2 p_xp_y
		+ a_1\frac{(xy+4)^2}{y^2}
		- a_2 \frac{xy-4}{\sqrt{xy}}
		-a_3 \frac{(xy+4)^2}{y\sqrt{xy}} +a_4 .
	\end{equation}
	It is then straightforward to compute
	\begin{equation}\label{eqn:S4.W}
		\mathcal{W} = \left\{
		\left(
		a_1\frac{1}{y^2}
		- a_2 \frac{xy-4}{(xy+4)^2\sqrt{xy}}
		-a_3 \frac{1}{y\sqrt{xy}} +\frac{a_4}{(xy+4)^2}
		\right){\rm d}x{\rm d}y
		\right\}/{\sim},
	\end{equation}
	where $\sim$ denotes the identification\footnote{In order not to overload notation, we tacitly adopt the convention that $(A/{\sim})=(A\setminus \{0\})/{\sim}$.}
	$g_1\sim g_2\Leftrightarrow g_1=kg_2$, $k\ne0$.
	The elements of $\mathcal{W}$ are (almost everywhere) non-singular 2-tensors up to multiplication by a non-zero constant.
\end{Example}

We shall now define a key object of the present paper -- the subvariety $\mathcal{Q}$ within $\mathcal{W}$ whose elements have vanishing Riemann curvature,
\[
\mathcal{Q} = \{ q\in\mathcal{W}\colon \mathrm{Riem}(q)=0 \} .
\]
By construction, this is invariant under St\"ackel transformations.
\begin{Remark}\quad
\begin{enumerate}\itemsep=0pt
\item[(i)] For elements $q\in\mathcal{W}$, the vanishing of their curvature tensor $\mathrm{Riem}(q)$ is indeed independent of the choice of representative for $q$.
\item[(ii)] Requiring the elements of $\mathcal{Q}$ to have vanishing curvature is a choice. Another choice would be, for instance, to consider the subset of elements $q\in\mathcal{W}$ such that the sectional curvature is constant,\footnote{However, note that requiring a specific constant would not be an invariant criterion.} i.e., elements~$q$ that live on the complex sphere. However, the (non-degenerate) St\"ackel classes (3,0), (0,11), (21,0) are shown in \cite{Kress07} to \emph{not} admit realisations on the complex 2-sphere. Therefore such a choice would not lead to an unambiguous characterisation of the St\"ackel classes considered here.
\end{enumerate}
\end{Remark}

The definition of $\mathcal{Q}$ is possible in any dimension $n$. We now restrict to $n=2$.
\begin{Lemma}In dimension~$2$, the variety $\mathcal{Q}\subset\mathcal{W}$ is defined by one homogeneous quadratic polynomial equation whose unknowns are the parameters of $V^{\mathcal{U}}$. The coefficients of this equation depend on the position $x\in M$ on the underlying manifold $M$.
\end{Lemma}

\begin{proof}In 2D the Riemannian curvature tensor is determined by its (unique) sectional curvature or, alternatively, by its scalar curvature.
	Moreover, in suitable local coordinates, any 2D metric can be written as $g=\phi^2{\rm d}x{\rm d}y$, such that the requirement of vanishing Riemannian curvature becomes
	\begin{equation}\label{eq:quadric.equation}
		V^{\mathcal{U}} V^{\mathcal{U}}_{xy} \phi^2 + 2\big(V^{\mathcal{U}}\big)^2 \phi_{xy}\phi - V^{\mathcal{U}}_xV^{\mathcal{U}}_y\phi^2 - 2\big(V^{\mathcal{U}}\big)^2 \phi_x\phi_y = 0 ,
	\end{equation}
	where we recall that $V^{\mathcal{U}}$ is a parametrisation of $\mathcal{U}$; the subscripts $x$, $y$ denote usual derivatives. Therefore~\eqref{eq:quadric.equation} is homogeneously quadratic in the $n+2=4$ parameters of the potential~$V^{\mathcal{U}}$, with coefficients depending on the position.
\end{proof}

\begin{Example}
	For the Hamiltonian of the Harmonic Oscillator,
	\[
	H = p_xp_y + a_3 xy +a_2 (x+y) + a_1 (x-y) +a_0 ,
	\]
	we find
	\[
	\mathcal{Q} = \big\{
	\big(a_3 xy +a_2 (x+y) + a_1 (x-y) +a_0
	\big) {\rm d}x{\rm d}y\colon  a_1^2-a_2^2=a_0a_3	\big\} .
	\]
	This is special since the condition on the parameters $a_i$ does not depend on the position $(x,y)$. In general the restriction will depend on the position (see below).
\end{Example}

While here the main concern is about \emph{non-degenerate} second-order 2D superintegrable systems, it may be pointed out that our definitions are not restricted to the non-degenerate setting.

\begin{Example}The Kepler--Coulomb system has the Hamiltonian $H=p_xp_y+\frac{a_1}{xy}+a_0$.
	We find, where $a_0 xy + a_1\ne0$, that
	\[
	\mathcal{F} = \left\{ \left(\frac{a_1}{xy}+a_0\right) {\rm d}x{\rm d}y\colon a_0a_1=0 \right\}/{\sim}
	\]
	consisting of two distinct projective points.
	Note however that Theorem~\ref{thm:classes} below does not apply to such degenerate systems (and neither does~\cite{Kress07}, which however contains one exceptional case).
\end{Example}

\section{Results}
We implement the method set out in the previous section for all cases of the classification \cite{Kalnins&Kress&Pogosyan&Miller}, and conclude that $\mathcal{Q}$ carries enough information to identify the St\"ackel class of a non-degenerate system.
For this purpose we introduce the following geometric object:
\begin{Definition}
	Let $g$ be a 2-dimensional metric (of arbitrary curvature) that admits a non-degenerate second-order superintegrable potential $V$.
	Denote by $\mathcal{S}=\mathcal{S}(g,V)$ the St\"ackel class of the system defined by the Hamiltonian $H=g^{ij}p_ip_j+V$, i.e., the set of all (non-degenerate second-order) superintegrable systems $h\in\mathcal{S}$ equivalent to the initial one under St\"ackel transform. We call
	\[
	\mathcal{F}(\mathcal{S}) = \{ h\in\mathcal{S}\colon \mathrm{Riem}(g_h)=0 \}/{\sim}
	\]
	the space of \emph{flat realisations} of~$\mathcal{S}$. Here~$g_h$ is the metric underlying $h\in\mathcal{S}$, and $\sim$ stands for identification under multiplication with a non-zero constant factor (i.e., we work projectively). Note that $g_h$ is, in particular, a metric conformally equivalent to the initial metric~$g$.
\end{Definition}

As mentioned earlier, the flat realisations of any (non-degenerate) St\"ackel class $\mathcal{S}$ are known from~\cite{Kress07}.
\begin{Lemma}\label{la:F.is.intersect.Q}
	The space $\mathcal{F}(\mathcal{S})$ is isomorphic to the intersection $\bigcap_N\mathcal{Q}$ over a neighborhood $N\subset M$.
\end{Lemma}

Note: We do not claim that the finding of Lemma~\ref{la:F.is.intersect.Q} is new. Implicitly the statement is found in the literature, e.g., \cite{Kalnins&Kress&Miller-II,Post10}. The author is not aware of any reference discussing the invariance of $\mathcal{Q}$, or its associated quadrics, however.
\begin{proof}
	By a suitable choice of coordinates, we can put any flat realisation of a St\"ackel class~$\mathcal{S}$ into its normal form [Em]; we follow the terminology of \cite{Kalnins&Kress&Pogosyan&Miller}.\footnote{We denote the 2D Harmonic Oscillator system by [E3], following \cite{Kalnins&Kress&Pogosyan&Miller}, where however it is not written with all $n+2$ parameters. The full potential is given in~\cite{Kress07}, for example, where the full system is distinguished by a~prime,~[E3$'$].}
	The Hamiltonian $H_m$ of [Em] then has the following form (note that only non-degenerate systems are listed):
	\begin{gather}\begin{split}
&H_1 = p_xp_y + a_3 xy+\tfrac{a_2}{(x+y)^2}+\tfrac{a_1}{(x-y)^2} +a_0, \\
&H_2 = p_xp_y + a_3 \big(4(x+y)^2-(x-y)^2\big)
			+a_2 (x+y) +\tfrac{a_1}{(x-y)^2} +a_0, \\
&H_3 = p_xp_y +a_3 xy +a_2x +a_1y +a_0, \\
&H_7 = p_xp_y + a_3 xy			+ a_2 \tfrac{y}{\sqrt{y^2-c^2}}
			+ a_1 \tfrac{x}{\sqrt{y^2-c^2} (y+\sqrt{y^2-c^2})^2}	+ a_0, \\
&H_8 = p_xp_y + \tfrac{a_3x}{y^3}			+\tfrac{a_2}{y^2}	+a_1 xy +a_0, \\
&H_9 = p_xp_y			+\tfrac{a_3}{\sqrt{y}}
			+a_2 (x+y)+\tfrac{a_1(x+3y)}{\sqrt{y}}+a_0, \\
&H_{10} = p_xp_y			a_3y +a_2 \big(x-\tfrac32 y^2\big)
			+a_1 \big(xy-\tfrac12 y^3\big) +a_0, \\
&H_{11} = p_xp_y
			+a_3x +\tfrac{a_2x}{\sqrt{y}}
			+\tfrac{a_1}{\sqrt{y}} +a_0, \\
&H_{16} = p_xp_y
			+\tfrac1{\sqrt{xy}} \left(
			a_3 +\tfrac{a_2}{\tfrac12(x+y)+\sqrt{xy}}
			+\tfrac{a_1}{\tfrac12(x+y)-\sqrt{xy}}
			\right) +a_0, \\
&H_{17} = p_xp_y
			+\tfrac{a_3}{\sqrt{xy}}
			+\tfrac{a_2}{x^2}
			+\tfrac{a_1}{x\sqrt{xy}} +a_0, \\
&H_{19} = p_xp_y
			+\tfrac{a_3y}{\sqrt{y^2-4}}
			+\tfrac{a_2}{\sqrt{x(y+2)}}
			+\tfrac{a_1}{\sqrt{x(y-2)}}
			+a_0, \\
&H_{20} = p_xp_y
			+\tfrac{1}{\sqrt{xy}} \left(
			a_3 +a_2 \sqrt{\tfrac12(x+y)+\sqrt{xy}}
			+a_1 \sqrt{\tfrac12(x+y)-\sqrt{xy}}
			\right) + a_0.\end{split}\label{eq:hamiltonians.Hm}
	\end{gather}
	For each of these, the position-dependent quadric $Q=Q_{(x,y)}(a_3,a_2,a_1,a_0)$ defining $\mathcal{Q}$ is polynomial in the coordinates or at least a sum of almost everywhere linearly independent functions $f_j(x,y)$.
	Thus, $Q=\sum_j\beta_jf_j(x,y)$ is a polynomial with constant coefficients $\beta_j$ depending quadratically on the parameters of $V^{\mathcal{U}}$. This proves the claim.
\end{proof}

With Observation~\ref{obs:main} and Lemma~\ref{la:F.is.intersect.Q} at hand, we are now ready to prove the main result. Theorem~\ref{thm:classes} below provides a tool for determining the St\"ackel class of any non-degenerate 2D second-order superintegrable system.

The resulting technique is quite efficient (see examples below). Its proof, however, is rather simple. In this respect it has to be acknowledged that the second-order superintegrable systems in 2D are already classified \cite{Kalnins&Kress&Pogosyan&Miller} and that we already know the complete list of St\"ackel classes appearing for 2D non-degenerate systems \cite{Kress07}, too. As a result, we can simply check \eqref{eq:quadric.equation} for a subset of the known normal forms of systems on flat space. Of course, this simplicity of the proof is no impediment to the strength or scope of the resulting method.\medskip

We are now going to see that two different St\"ackel classes give rise to different spaces~$\mathcal{F}$. Therefore, from the knowledge of~$\mathcal{F}$, we can infer the St\"ackel class of a non-degenerate second-order superintegrable system.
As an explicit example, take the Hamiltonian $H_8 = p_xp_y + V[a_3,a_2,a_1,a_0]$ from~\cite{Kalnins&Kress&Pogosyan&Miller} (see the list~\eqref{eq:hamiltonians.Hm} above). Here, $V^\mathcal{U}=V[a_3,a_2,a_1,a_0]$ is a~concrete parametrisation. Since we work projectively, the potential is to be considered modulo multiplication by an irrelevant constant, which we denote as $V[a_3:a_2:a_1:a_0]$.
We find
\begin{align*}
	\mathcal{Q}
	&= \left\{ V[a_3:a_2:a_1:a_0] {\rm d}x{\rm d}y\colon
	\frac{\big(a_0a_3 y^6+3a_2a_3 y^4-3a_0a_1 y^2-a_1a_2\big)y^3}
	{a_3xy^4+a_0y^3+a_1x+a_2y} = 0 \right\} \\
	\intertext{and thus}
	\mathcal{F}
	&= \{ V[a_3:a_2:a_1:a_0] {\rm d}x{\rm d}y \colon
	a_0a_3=0,\, a_2a_3=0,\, a_0a_1=0,\, a_1a_2=0 \},
\end{align*}
wherever $a_3xy^4+a_0y^3+a_1x+a_2y\ne0$.
One therefore obtains a space with two distinct components,
\[
\mathcal{F}
= \{ V[a_3:a_2:a_1:a_0] {\rm d}x{\rm d}y\colon a_2=0=a_0 \}
\cup \{ V[a_3:a_2:a_1:a_0] {\rm d}x{\rm d}y\colon a_1=0=a_3 \}
\]
($x\ne0$). Table~\ref{tab:classes} summarises the results of the analogous computations for all metrics in the list~\eqref{eq:hamiltonians.Hm}.

\begin{table}[t]\centering
	\caption{The table details the non-degenerate St\"ackel equivalence classes in dimension 2, their flat realisations and the associated quadric $\mathcal{F}$. Each row stands for a St\"ackel class~$\mathcal{S}$, specified in the first column and labeled as in \cite{Kress07}. The second column specifies, within this class, its flat realisations $h\in\mathcal{S}$ up to isometries (i.e., up to coordinate changes). These realisations~$h$ are denoted with labels as in \cite{Kalnins&Kress&Pogosyan&Miller}, cf.\ also the above list. The third column describes the variety $\mathcal{F}$ for the respective St\"ackel class; note that it can be parametrised in various equivalent ways, but the varieties themselves are invariant and characteristic to each class.}\label{tab:classes}
	\vspace{1mm}

	\begin{tabular}{l|ll}
		\toprule
{St\"ackel class $\mathcal{S}$}
		& {Flat systems in $\mathcal{S}$}
		& {Description of $\mathcal{F}(\mathcal{S})$}\\
		up to St\"ackel transform
		& up to isometries
		& as a variety
		\\ \midrule
		(0,11)
		& E3, E11, E20
		& quadric surface $uv=a^2+b^2$
		\\
		(21,0)
		& E7, E8, E17, E19
		& two projective lines
		\\
		(3,11)
		& E9, E10
		& one projective line
		\\
		(21,2)
		& E1, E16 
		& two projective points
		\\
		(3,2)
		& E2 
		& one projective point
		\\
		(111,11)
		& none 
		& empty
		\\ \bottomrule
	\end{tabular}
\end{table}

\begin{Theorem}\label{thm:classes}
	The variety $\mathcal{F}$ determines the St\"ackel class of a second-order non-degenerate superintegrable system in dimension~$2$ unambiguously.
\end{Theorem}
\begin{proof}Any second-order superintegrable system in 2D is St\"ackel equivalent to a superintegrable system on a constant curvature space \cite{Capel_phdthesis,Kalnins&Kress&Miller-II}. Therefore Table~\ref{tab:classes} covers all cases and we immediately infer the asserted statement.
\end{proof}

With the help of this theorem, we can easily determine the St\"ackel class for a 2-dimensional non-degenerate second-order superintegrable system. We illustrate this with two examples of non-zero curvature.

\begin{Example}
As a first example we consider the system [S4], defined on the 2-sphere with the metric~\eqref{eqn:S4.metric} and the Hamiltonian~\eqref{eqn:S4.hamilton}. Using~\eqref{eqn:S4.W}, we find that
\[
	\mathcal{F} = \left\{
	\left(
	\frac{a_1}{y}
	-\frac{a_3}{\sqrt{xy}}
	\right)\frac{{\rm d}x{\rm d}y}{y}\colon \text{either $a_1=0$ or $a_3=0$, but not both}
	\right\}/{\sim},
\]
where we require $xy+4\ne0$, and either $y\ne0$ or $x\ne0$, respectively.
	In other words, $\mathcal{F}$ consists of two distinct projective points.
	According to Table~\ref{tab:classes}, this is the St\"ackel class~(21,2). Indeed, the system~[S4] -- defined on the 2-sphere -- is St\"ackel equivalent to the systems~[E1] and~[E16], i.e., the Hamiltonians $H_1$ and $H_{16}$ from~\eqref{eq:hamiltonians.Hm}, defined on flat 2-space~\cite{Kress07}.
\end{Example}

\begin{Example}\label{ex:matveev}
	In \cite{Vollmer_2020}, see also \cite{MV2019,MV2020}, it is proven that any non-degenerate second-order superintegrable system that admits one unique, essential projective symmetry\footnote{A vector field is called projective if its local flow maps geodesics into geodesics if we disregard parametrisation. It is called essential if it is not a homothety. The word ``unique'' here means that the projective symmetry algebra is 1-dimensional.}
	is projectively equivalent to\footnote{Note that \eqref{eqn:matveev.g} itself has a 1-dimensional projective algebra, but its projective vector fields are homothetic. Nonetheless, any 2D second-order superintegrable metric with a unique, essential projective symmetry is projectively equivalent to~\eqref{eqn:matveev.g}.}
	the Hamiltonian
	\begin{equation}\label{eqn:matveev.H}
		H = \frac{p_xp_y}{\big(x+y^2\big)}
		- a_3 \frac{y\big(y^2+3x\big)}{x+y^2}
		+ a_2 \frac{y}{x+y^2}
		+ \frac{a_1}{x+y^2} +a_0 .
	\end{equation}
	The underlying metric is
	\begin{equation}\label{eqn:matveev.g}
		g = \big(x+y^2\big){\rm d}x{\rm d}y ,
	\end{equation}
	see for instance \cite{bolsinov_2009,dini_1869, matveev_2012}.
	We find
	\begin{equation}\label{eqn:matveev}
		\mathcal{F} = \{ (a_1+a_2y) {\rm d}x{\rm d}y \ne 0 \}/{\sim},
	\end{equation}
	where $y\ne0$.
	Therefore, we have confirmed that $\mathcal{F}$ is a projective line, and thus $H=g^{ij}p_ip_j+V$ belongs to the St\"ackel class~(3,11).
\end{Example}

\looseness=-1 We conclude this section with a closer look at the inner structure of the invariant variety~$\mathcal{F}$.
By definition, the invariant variety $\mathcal{F}$ is the space of all flat realisations of a given St\"ackel class. While before we were using~$\mathcal{F}$ in order to characterise the St\"ackel class, we shall now take a different stance: We aim to understand the internal structure of~$\mathcal{F}$ as a space of flat superintegrable systems. In other words, for each row of Table~\ref{tab:classes}, i.e., for each St\"ackel class, we ask how the individual flat systems listed in the second column are ``placed'' within the variety~$\mathcal{F}$.

The strategy for this goal is rather simple: We compute $\mathcal{U}g$ for one realisation $H=g^{ij}p_ip_j+V^{\mathcal{U}}$ from each row in Table~\ref{tab:classes}. Due to the invariance of~$\mathcal{F}$ it does not matter which actual realisation we select for the computation.
For convenience, we may chose coordinates $(x,y)$ such that $g={\rm d}x{\rm d}y$. Then the isometry operations are $x\to\lambda x+a_1$ and $y\to\frac{y}{\lambda}+a_2$; the constants~$a_1$,~$a_2$ and $\lambda\ne0$ have to be adjusted such that we arrive at the suitable normal form from~\cite{Kalnins&Kress&Pogosyan&Miller}.

\begin{Example}\label{ex:class.3,11}
	For the class (3,11) we can use~\eqref{eqn:matveev}. Since either $a_1$ or $a_3$ has to be non-zero, we arrive at two distinct cases:
	\begin{enumerate}\itemsep=0pt
		\item Case 1: $a_3=0$, but $a_1\ne0$, i.e., w.l.o.g.\ $a_1=1$. Together with $a_2=0=a_4$ we arrive at the St\"ackel transform with conformal factor $U=\frac{1}{x+y^2}$. The transformed Hamiltonian is therefore
		\begin{align*}
			H' &= \frac{H}{U} = p_xp_y +a_1
			- \tfrac{a_2}{2} y \big(y^2-3x\big) + a_3y + a_4 \big(y^2+x\big) \\
			&= p_{x'}p_{y'} + a_1'y' + a_2' \big(x'-\tfrac32(y')^2\big)
			+ a_3 y' \big(x'-\tfrac12(y')^2\big) + a_4' ,
		\end{align*}
		after a suitable change of coordinates from $(x,y)$ to $(x',y')$, and redefinition of the parameters from the $a_i$ to the $a_i'$.
		This is the system [E10] of~\cite{Kalnins&Kress&Pogosyan&Miller}.
		\item Case 2: $a_3\ne0$. In an analogous way, one arrives to the system [E9] of~\cite{Kalnins&Kress&Pogosyan&Miller}.
	\end{enumerate}
\end{Example}

Continuing the procedure as in Example~\ref{ex:class.3,11} for the other St\"ackel classes, we find the following.

\begin{Proposition}\label{prop:inner.structure}
	For the class $(111,11)$, $\mathcal{F}$ is empty, $\mathcal{F}=\varnothing$, i.e., no flat realisations exist. For the other St\"ackel classes, we have:	
	\begin{description}
		\item[Class (3,2).]\itemsep=0pt
		The variety $\mathcal{F}$ consists of one projective point, corresponding to the system~{\rm [E2]} of~{\rm \cite{Kalnins&Kress&Pogosyan&Miller}}.
		\item[Class (21,2).]
		The variety $\mathcal{F}$ consists of two projective points, one corresponding to {\rm [E1]}, the other to~{\rm [E16]}.
		\item[Class (3,11).]
		This is the first non-trivial case. $\mathcal{F}$ consists of one projective line, which generically is of type~{\rm [E9]}. The system~{\rm [E10]} corresponds to a projective point lying within this projective line.
		\item[Class (21,0).]
		Two disjoined projective lines are contained in $\mathcal{F}$. One line is generically~{\rm [E19]}, containing one point that is~{\rm [E17]}. The other line is {\rm [E7]} generically and contains a point that is~{\rm [E8]}.
		\item[Class (0,11).]
		The most interesting variety contains the Harmonic Oscillator and is governed by the position-independent quadric $a^2+b^2=uv$ for $a,b,u,v\in\mathbb{C}$. In the quadric defining $\mathcal{Q}$ the position dependent contributions factor out. Somewhat surprisingly, the system~{\rm [E20]} is realised generically, when $\mathcal{F}$ is described by $v=\frac{a^2+b^2}{u}$ $(u\ne0)$. The system~{\rm [E11]} is realised if $u=0$, $a\ne0$ and $b\ne0$. The quadric is $a^2+b^2=(a+{\rm i}b)(a-{\rm i}b)=0$. The projective point with $u=a=b=0$ realises~{\rm [E3]}.
	\end{description}
\end{Proposition}

Note that for class (111,11) there exist no flat realisations. There are, however, three distinct non-degenerate systems with constant sectional curvature, called [S7], [S8] and [S9] in~\cite{Kalnins&Kress&Pogosyan&Miller}, which realise (111,11) on the complex 2-sphere \cite{Kress07}. Analogously to $\mathcal{F}$, one could now define a variety $\mathcal{E}=\{q\in\mathcal{W}\colon  \mathrm{Riem}(g)=\text{const}\ne0\}$ and ask for its inner structure analogously to Proposition~\ref{prop:inner.structure}. For instance, for the class (111,11), one obtains a variety $\mathcal{E}$ whose complete primary decomposition leads to 11 components. A~further study of these components does however not appear to reveal any new, interesting facts relating to the purposes of the present paper. Hence we refrain from any further analysis of~$\mathcal{E}$.

\section{Discussion and generalisations}

\subsection{Comparison with established method}

It is worthwhile contrasting Theorem~\ref{thm:classes} with the method established by \cite{Kress07} that appears to be the only preexisting technique for this task available for immediate use.
Essentially, one has to complete the following steps:
\begin{enumerate}\itemsep=0pt
	\item Compute functionally independent integrals $F_1$, $F_2$ explicitly for the Hamiltonian~$H$.
	\item Compute their Poisson bracket $R=\{F_1,F_2\}$.
	\item Rewrite $R^2$ (a~sextic polynomial w.r.t.\ momenta) as a cubic in $H$, $F_1$, $F_2$.
	\item Bring this cubic into normal form, cf.~\cite{Kress07}.
\end{enumerate}
Note that there exist different choices for $F_1$ and $F_2$, which lead to different but proportional results for $R$. This ambiguity is resolved only in the last step by reverting to the normal forms.
It is instructive to look at an example:
\begin{Example}
	We reconsider the system of Example~\ref{ex:matveev}. Its Hamiltonian \eqref{eqn:matveev.H} is defined on a~mani\-fold of non-constant sectional curvature and its potential is non-degenerate as it has four linear parameters with independent functions as coefficients.
	We follow the four aforementioned steps:
\begin{enumerate}\itemsep=0pt
\item[(i)] The Hamiltonian \eqref{eqn:matveev.H} has the advantage that we do not have to integrate for the integrals of motion. In fact, two metrics that are projectively equivalent to \eqref{eqn:matveev.g} are given in \cite{matveev_2012}, and from these we can obtain the following two integrals of motion by using the formula from \cite{topalov_2003}, see also~\cite{Vollmer_2020},
	\begin{gather*}
		F_1 =
		\left( p_x^2 - \frac{2y p_xp_y}{y^2+x} \right)
		- \frac{a_3 y^2\big(y^2-3x\big)}{y^2+x}
		+ \frac{2a_2 y^2}{y^2+x} + \frac{2a_1 y}{y^2+x} + a_0,	\\
		F_2 =
		\left( 12x p_x^2 -4y \frac{y^2+9x}{y^2+x} p_xp_y +9 p_y^2 \right)
		\\
\hphantom{F_2 =}{}
		+\frac{a_3 \big(y^2-3x\big)^3}{y^2+x}
		+ \frac{2a_2 \big(y^2-3x\big)^2}{y^2+x}
		+ \left(\frac{8 y\big(y^2-3x\big)}{y^2+x} -8\right) a_1 + a_0.
	\end{gather*}
\item[(ii)] We compute
	\begin{gather*}
		R = -12p_x \left(
		2 p_x^2 +\frac{3 p_y^2}{x+y^2} -\frac{6y p_xp_y}{x+y^2}
		\right)	\\
\hphantom{R =}{}
		+36p_x \frac{
			a_3\big(y^4+3x^2\big)
			+a_2 \big(y^2-x\big)
			+2a_1y}{x+y^2}		-36p_y \frac{2a_3y^3 +a_2y +a_1}{x+y^2}.
	\end{gather*}
\item[(iii)] Its square $R^2$ can be written as
	\begin{gather*}
		R^2 = 576 F_1^3
		+ 1728 (a_0+a_2) F_1^2 -432 a3 F_1F_2
		\\
\hphantom{R^2 =}{} +\big( {-}1728a_1 H
		+ 432(4a_0^2+4(a_0-2a_3)a_1+8a_0a_2+3a_2^2+a_0a_3)
		\big) F_1
		\\
\hphantom{R^2 =}{}  +\big( 144 H^2 -288a_0 H
		+ 144\big(a_0^2 - 3a_0a_3\big)\big) F_2	- 144 (a_0-8a_1)H^2
		\\
\hphantom{R^2 =}{}  + 288 (a_0^2 -(14a_0 +9a_2)a_1)H
		+ 288 \big(10a_0^2 +9a_0a_2 -12a_0a_3\big)a_1		\\
\hphantom{R^2 =}{} +432 a_0^3 +1728a_0^2a_2 +1296a_0a_2^2		+ 432a_0^2a_3 +3888a_1^2a_3.
	\end{gather*}
\item[(iv)] After a linear redefinition $F_1'=\sqrt[3]{576} (F_1+a_0+a_2)$ and a rescaling $F_2'=-36\sqrt[3]{3}F_2$, we find
	\[
	R^2 = (F_1')^3 + a_3 F_1'F_2' + \mathcal{O} ,
	\]
where $\mathcal{O}$ is at most linear in $(F_1',F_2')$. Therefore, according to~\cite{Kress07}, the system is of type~(3,11). This is consistent with what we found in Example~\ref{ex:matveev}.
\end{enumerate}
\end{Example}

Note that we have needed the explicit expressions for the integrals, and for their Poisson bracket.
This makes the procedure typically computationally more intensive than the one presented in Theorem~\ref{thm:classes}.

\subsection{Conformal superintegrability}\label{sec:conformal}

In the remainder of this section we comment on two possible generalisations of Theorem~\ref{thm:classes}. The first is the extension of our results to conformal superintegrability, the other one to dimensions higher than~2.

Let us begin with a remark on conformal superintegrability. In this case, integrals of motion~\eqref{eq:integral} are replaced by conformal integrals:
A \emph{second-order conformal integral} is a function $F(x,p)=K^{ij}(x)p_ip_j+W(x)$ such that
\begin{equation}\label{eq:conformal.integral}
	\{F,H\} = \omega H
\end{equation}
holds, for some polynomial $\omega=\omega^i(x)p_i$ linear in momenta. Obviously, every integral of $H$ -- satisfying \eqref{eq:integral}~-- is also a conformal integral for $H$, with $\omega=0$.
With \eqref{eq:conformal.integral} instead of \eqref{eq:integral}, the components $K_{ij}$ are now the components of a conformal Killing tensor. As a consequence, instead of~\eqref{eq:dKdV}, we obtain the equation \cite{Kalnins&Kress&Miller-II,Kalnins&Kress&Miller&Post11,KSV_2020}
\[
\nabla_{[i}K\indices{_{j]}^a}\nabla_aV
= K\indices{^a_{[i}}\nabla^2_{j]a}V
+\omega_{[i}V_{,j]} + \omega_{[i,j]}V .
\]
Conformal \looseness=-1 superintegrability can then be defined analogously to proper superintegrability, but we require functionally independent \emph{conformal} integrals~\eqref{eq:conformal.integral} instead of proper ones, i.e., instead of~\eqref{eq:integral}.
Non-degeneracy is then also defined analogously, but one thing should be noted: For non-degenerate properly superintegrable systems, we have $V^\mathcal{U}=a_0+\sum_{k=1}^{n+1} a_kf_k$ for $(n+1)$-many functions $f_k$. In conformal systems, however, we usually have $n+2$ non-constant functions~$f_k$, $V^\mathcal{U}=\sum_{k=0}^{n+1}a_kf_k$. In both cases, the admissible potentials form a $(n+2)$-dimensional linear space~$\mathcal{U}$.
Finally, the \emph{conformal equivalence} of conformally superintegrable systems is defined similar to St\"ackel equivalence, see, e.g.,~\cite{BKM1986,Kalnins&Kress&Miller-II, KSV_2020}.
For our purposes, it is enough to allow, in~\eqref{eq:staeckel}, arbitrary functions~$U$, which not necessarily are compatible potentials of~$H$. In other words, we allow $U\not\in\mathcal{U}$ in addition.
With these generalisations of superintegrability at hand, it is easy to extend Theorem~\ref{thm:classes} to non-degenerate second-order 2D \emph{conformally} superintegrable systems.

\begin{Theorem}\label{thm:conformal}
The variety $\mathcal{F}$ determines the St\"ackel class of a second-order non-degenerate conformally superintegrable system in dimension~$2$ unambiguously.
\end{Theorem}
\begin{proof}It is shown in \cite{Capel_phdthesis} that any conformally superintegrable system is St\"ackel equivalent to a~properly superintegrable system, i.e., no new St\"ackel classes can appear.
	The assertion is then easily confirmed by following the individual steps in the proof of Theorem~\ref{thm:classes}.
	The key point in this respect is that Observation~\ref{obs:main} continues to hold true if~$U\not\in\mathcal{U}$.
\end{proof}

\subsection{Higher dimension}

The current paper focuses on dimension $n=2$. In dimensions $n>2$ it is still possible, by the same reasoning, to construct the invariant $\mathcal{Q}\subset\mathcal{W}$, but it is generally not described by a~single quadric.
Also, in higher dimension it is not possible to use $\mathcal{F}$ in the same way as in 2D for a~characterisation of the St\"ackel class. In this context, it seems appropriate to mention that also the method of~\cite{Kress07} is restricted to dimension~2.
Indeed, one quickly finds that $\mathcal{Q}$ cannot be used to identify the St\"ackel class in dimension~3 already. In the 3D case, the explicit normal forms of second-order non-degenerate superintegrable systems \cite{Capel_phdthesis, KKM06a,KKM07c} facilitate the computation; however, in most cases, $\mathcal{F}$ is just one projective point.
Such ambiguities should be expected in any higher dimension, pointing to fundamental structural particularities of 2D second-order superintegrable systems compared to higher dimensions \cite{Kress&Schoebel,Kress&Schoebel&Vollmer,KSV_2020}.

\subsection*{Acknowledgements}
The author is indebted to Jonathan Kress, Joshua Capel and Konrad Sch\"obel for many discussions on superintegrability, St\"ackel transforms and numerous related topics, and thanks the anonymous referees for contributing valuable suggestions that lead to improvements of this paper. A special thank you goes to Konrad Sch\"obel for helpful comments on the manuscript.

This paper was principally written when the author was a fellow of the German Research Foundation~-- Deutsche Forschungsgemeinschaft (DFG): Andreas Vollmer acknowledges funding from a DFG research fellowship with the project number 353063958 as well as through a subsequent return fellowship. Andreas Vollmer thanks the University of Stuttgart and the University of New South Wales for hospitality.

\vspace{-2mm}

\pdfbookmark[1]{References}{ref}
\LastPageEnding

\end{document}